%% file: StochOptArxiv.tex
\documentclass[11pt]{article}
\input{preamble-new}

\input{spage}

\usepackage{graphicx}
\usepackage{latexsym}
\usepackage{algorithmic}
\usepackage[colorlinks]{hyperref}

\newcommand{\E}{{\mbox {\bf E}}}

\newcommand{\C}{\mathcal{C}}
\newcommand{\F}{\mathcal{F}}
\newcommand{\rev}[1]{{#1} }
\renewcommand{\S}{\rev{Section}}

\begin{document}

\title{Approximation algorithms for stochastic and risk-averse
optimization\thanks{An earlier version of this paper appeared as \cite{Srinivasan07}.}}
\author{Jaroslaw Byrka\thanks{Institute of Computer Science, University of Wroclaw, Joliot-Curie 15, 50-383 Wroclaw, Poland. {\tt jby@ii.uni.wroc.pl}
}~~and Aravind Srinivasan\thanks{Department of Computer Science and
Institute for Advanced Computer Studies, University of Maryland,
College Park, MD 20742, USA. \texttt{srin@cs.umd.edu}}}
\date{}
\maketitle

\begin{abstract}
We present improved approximation algorithms in stochastic optimization.
We prove that the multi-stage stochastic versions of
covering integer programs (such as set cover and vertex
cover) admit essentially the same approximation algorithms as their
standard (non-stochastic) counterparts; this improves upon 
work of Swamy \& Shmoys which shows an approximability that depends
multiplicatively on the number of stages. We also present 
approximation algorithms for facility location and some of its
variants in the $2$-stage recourse model, improving on previous
approximation guarantees. 
We give a $2.2975$-approximation algorithm in the standard polynomial-scenario model
and an algorithm with an expected per-scenario $2.4957$-approximation guarantee, which is applicable
to the more general black-box distribution model. 

\end{abstract}

\section{Introduction}
\label{sec:intro}
Stochastic optimization attempts to model uncertainty in the
input data via probabilistic
modeling of future information. It originated in the work
of Beale \cite{Beale55} and Dantzig \cite{Dantzig55} five
decades ago, and has found application in several areas
of optimization. There has been a flurry of
algorithmic activity over the last decade in this field, especially from the viewpoint of
approximation algorithms; see
the survey \cite{SwamyS06} for a thorough discussion of
this area. 

In this work, we present improved approximation algorithms for
various basic problems in stochastic optimization. We start
by recalling the widely-used \textit{$2$-stage recourse model}
\cite{SwamyS06}. 
Information about the input instance is revealed in two stages here. 
In the first, we
are given access to a distribution $D$ over possible realizations
of future data, each such realization called a \textit{scenario};
given $D$, we can commit to an anticipatory part $x$ of the total 
solution, which costs us $c(x)$. In the second stage, a scenario
$A$ is sampled from $D$ and given to us, specifying the complete
instance. We may then augment $x$ by taking recourse actions
$y_A$ that cost us the additional amount of 
$f_A(x,y_A)$ in order to construct a feasible 
solution for the
complete instance. The algorithmic goal is to construct 
$x$ efficiently, as well as $y_A$ efficiently (precomputed
for all $A$ if possible, or computed when $A$ is revealed to us), in order
to minimize the total expected cost, $c(x) + \E_A[f_A(x,y_A)]$.
(In the case of randomized algorithms, we further take the
expectation over the random choices of the algorithm.) This is
the basic cost-model. We will also study ``risk-averse'' relatives
of this expectation-minimization version. 
There is a natural extension of the above to $k \geq 2$ stages;
see \cite{SwamyS05} for a nice motivating example for the case
where $k > 2$, and for the precise model. We present just those
details of this model that are relevant for our discussion, in
\S~\ref{sec:cip}. 

As an example, the $2$-stage version of set cover is as follows.
As usual, we have a finite ground set $X$ and a given family
of subsets $S_1, S_2, \ldots, S_m$ of $X$; the stochasticity
comes from the fact that the actual set of elements to be covered
could be a \textit{subset} of $X$, about which we only
have probabilistic information. As above, we can sample $D$ to
get an idea of this subset in stage I; we can also buy each $S_j$ for
some given cost $c_j$ in stage I. Of course, the catch is that
future costs are typically 
higher:
i.e., for all $j$ and $A$,
the cost $c_{A,j}$ of buying $S_j$ under scenario $A$ in stage II, could
be much more than $c_j$. This reflects the increased cost of
rapid-response, as opposed to the advance provisioning of
any set $S_j$. As in set cover, a feasible solution is a
collection of $\{S_j\}_{j=1}^m$ that covers all of the finally-revealed elements
$A$. Thus, we will choose some collection $x$ of sets 
in stage I by using the distributional information about $A$,
and then augment $x$ by choosing further sets $S_j$ in stage II when
we know $A$. One basic goal is to minimize the total expected
cost of the two stages. 

How is $D$ specified? As mentioned in \cite{SwamyS06}, there has been 
recent work in algorithms where the data (e.g., demands) come from 
a product of \textit{independent},
explicitly-given distributions (see, e.g., the discussions in
\cite{MohringSU99,ImmorlicaKMM04,DeanGV04}). 
One major advantage here is that it
can succinctly capture even exponentially many scenarios. However,
just as in 
\cite{RaviS04,ShmoysS04,GuptaPRS11,SwamyS05,SwamyS06}, we are interested in dealing with correlations
that arise in the data (e.g., correlations due to geographic
proximity of clients), and hence will not deal with such \textit{independent
activation} models here. So, our general definition is where
we are given access to a black-box \rev{that} can generate samples
according to $D$. Alternatively, we could be explicitly given
the list of scenarios and their respective probabilities. In
this case, algorithms that run in time polynomial in the other
input parameters naturally require that the total number of
scenarios be polynomially bounded. A natural question to ask
is: can we ``reduce'' the former model to the latter, by taking
some polynomial number of scenarios from the black-box, and 
constructing an explicit list of scenarios using their
empirical probabilities? Indeed, this \textit{sample-average
approximation} method is widely used in practice: 
see, e.g., \cite{LinderothSW04,VerweijAKNS03}. The work of 
\cite{ShmoysS04,CharikarCP05,SoZY06} has shown that we can
essentially reduce the black-box model to the polynomial-scenario
model for the case of $k = 2$ stages,
by a careful usage of sampling for the problems
we study here: the error introduced by the sampling will
translate to a multiplicative $(1 + \epsilon)$ factor in
the approximation guarantee, where $\epsilon$ can be made
as small as any desired inverse-polynomial of the input size.
We will define the $k$-stage model using the
details relevant to us in \S~\ref{sec:cip}; in 
\S~\ref{sec:facloc-exp}, where we only deal with
$k = 2$ stages, we will present our algorithms in the polynomial-scenario model.
As it is discussed below, some of the presented algorithms are compatible
with the reduction from the black-box model and hence provide more general results.

Our results are as follows. We consider the $k$-stage model
in \S~\ref{sec:cip}; all the problems and results here,
as well as earlier work on these problems, is for arbitrary
constants $k$. The \textit{boosted sampling} approach of
\cite{GuptaPRS11} leads to approximation guarantees that are
exponential for $k$ for problems such as vertex cover (and
better approximations for the $k$-stage Steiner tree problem). This
was improved in \cite{SwamyS05}, leading to approximation guarantees
for vertex cover, set cover, and facility location that are
$k$ times their standard (non-stochastic) threshold:
for example, approximation guarantees of $2k + \epsilon$ and
$k \ln n$ for vertex cover and set cover respectively are developed
in \cite{SwamyS05}. Removing this dependence on $k$ is mentioned
as an open problem in \cite{SwamyS06}. We resolve this by
developing simple randomized approximation algorithms that yield,
for the family of covering integer programs, essentially the same
approximation algorithms as for their non-stochastic counterparts.
In particular, we get guarantees of $2 + \epsilon$ and $(1 + o(1)) \ln n$
respectively, for vertex cover and set cover. 
Except for a somewhat non-standard version of vertex cover studied in
\cite{RaviS04}, these are improvements even for the case of $k = 2$. 
Chaitanya Swamy (personal communication, June 2006) has informed us
that Kamesh Munagala has independently obtained the result for
$k$-stage vertex cover. 

Our next object of study is the classical facility location problem.
Recall that in the standard (non-stochastic) version of the facility
location problem, we are given
a set of clients $\mathcal{D}$ and a set of facilities $\mathcal{F}$.
The distance from client $j$ to facility $i$ is $c_{ij}$, and these values
form a metric. Given a cost $f_i$ for opening each facility $i$, we want 
to open some of the facilities, so that the sum of the opening costs and
the distances traveled by each client to its closest open facility, is
minimized. (As usual, all results and approximations translate without
any loss to the case where each client $i$ has a demand $d_i \geq 0$
indicating the number of customers at $i$, so we just discuss the case
$d_i \in \{0,1\}$ here.) Starting with \cite{ShmoysTA97}, there has been
a steady stream of constant-factor approximation algorithms for the
problem, drawing from and contributing to techniques in LP rounding,
primal-dual methods, local search, greedy algorithms etc. The current-best
lower and upper bounds on the problem's approximability are
$1.46 \rev{\ldots} $ \cite{GuhaK99} and $1.488$ \cite{Li11}.
The stochastic version of the facility location problem has also received 
a good deal of attention
in the case of $k = 2$ stages \cite{RaviS04,Swamy04,ShmoysS04}. 
Here, each facility $i$ can be opened at cost $f_i^I$ in stage I, and at
a cost $f_i^A$ when a scenario $A$ materializes in stage II; each
scenario $A$ is simply a subset of $\mathcal{D}$, indicating the actual
set of clients that need to be served under this scenario. 
The goal is to open some facilities
in stage I and then some in stage II; we develop improved approximation
algorithms in two settings here, as discussed next. 

Our first setting is the basic one of minimizing the total expected cost,
just as for covering problems. That is, we aim to minimize the expected
``client-connection costs + facility-opening cost'', where the expectation
is both over the random emergence of scenarios, and the internal random
choices of our algorithm. We first propose a $2.369$-approximation in
\S~\ref{sec:ufl-2stage-expmin}, improving upon the current best value of 2.976 implied by~\cite{Li11,ShmoysS04,SwamyS06}. Our approach here
crucially exploits an asymmetry in the facility-location algorithms of
\cite{JainMMSV02,JainMS02}. 
Next, in \S~\ref{sec:LP-rounding}, we give an LP-rounding algorithm that delivers solutions which expected cost  may be bounded
by $1.2707 \cdot C^* + 2.4061 \cdot F^*$, where $C^*$ and $F^*$
are, respectively, the connection and the facility opening costs of the initial fractional solution.
Finally, in \S~\ref{sec:ufl-combined}, we combine our two algorithms for stochastic facility location
to obtain an even better approximation guarantee. 
Namely, we prove that the better of the two solutions is expected to have cost
at most $2.2975$ times the cost of the optimal fractional solution.

Our second setting
involves an additional ingredient of ``risk-aversion'', various
facets of which have been studied in works including
\cite{GuptaRS07,DhamdhereGRS05,SoZY06}. The motivation here is
that a user may be risk-averse, and may not want to end up paying
much if (what was perceived to be) a low-probability or low-cost scenario
emerges: overall expectation-minimization alone may not suffice. Therefore,
as in \cite{GuptaRS07}, we aim for ``local'' algorithms: those where
for each scenario $A$, its expected final (rounded)
cost is at most some guaranteed (constant) factor 
times its fractional counterpart $Val_A$. Such a result is very useful
since it allows the inclusion of budgets to the various scenarios, 
and to either prove that these are infeasible, or to come within a constant factor
of each particular budget in expectation \cite{GuptaRS07}. 

In \S~\ref{sec:per-scenario} we give a new randomized rounding algorithm
and prove that it returns a solution with expected
cost in scenario $A$ is at most $2.4957$ times its fractional counterpart 
$Val_A = \sum_{i \in \F} (f_i^I y^*_i + f_i^A y^*_{A,i}) + \sum_{j\in A}\sum_{i\in \F} c_{ij}x^*_{A,ij}$,
improving on the $3.225 \cdot Val_A$ bound of \cite{Swamy04} and the $3.095 \cdot Val_A$
bound form an earlier version of this paper \cite{Srinivasan07}.
The algorithm is analized with respect to the fractional solution to the LP relaxation and the analysis
only uses a comparison with cost of parts of the primal solution, which makes the algorithm compatible with the reduction from the black-box model to the polynomial scenario model. 
Note, however, that the above mentioned budget constraints may only be inserted for explicitly speciffied scenarios.
In fact it was shown in~\cite{Swamy11} that obtaining budgeted risk-averse solutions in the black-box modes is not possible.

Finally, in \S\ref{sec:strict-per-scenario} we briefly discuss an even more risk-averse setting,
namely one with deterministic per-scenario guarantee. We note that the the algorithm form \S~\ref{sec:per-scenario}
actually also provides deterministic per-scenario upper bounds on the connection cost.

Thus, we present improved approximation algorithms in stochastic
optimization, both for two stages and multiple stages, based on
LP-rounding. 

\section{Multi-stage covering problems}
\label{sec:cip}

We show that general stochastic $k$-stage covering integer programs
(those with all coefficients being non-negative, and with the
variables $x_j$ allowed to be arbitrary non-negative integers)
admit essentially the same approximation ratios as their 
usual (non-stochastic)
counterparts. The model is that
we can ``buy'' each $x_j$ in $k$ stages: the final value of
$x_j$ is the sum of all values bought for it. 
We also show that $k$-stage vertex cover can be
approximated to within $2 + \epsilon$; similarly, $k$-stage
set cover with each element of the universe contained in at most
$b$ of the given sets, can be approximated to within $b + \epsilon$.

\noindent
\textbf{Model:} 
The general model for $k$-stage stochastic optimization
is that we have stochastic information about the world
in the form of a $k$-level tree with leaves containing complete information about the instance.
The revelation of the data to the algorithm corresponds to traversing
this tree from the root node toward the leaves. 
This traversal is a stochastic process
that is specified by a probability distritution (for the choice of child to move to)
defined for every internal node of the tree.
The algorithm makes irrevocable decisions
along this path traversal.

Since we consider only a constant number of stages $k$,
the polynomial-scenario estimation of a distribution accessable via a blackbox
is also possible: only that the degree of the polynomial depends linearly on $k$.
Given the results of \cite{SwamyS05}, solving a $k$-stage covering integer program (CIP)
reduces to solving polynomial-sized tree-scenario problems (for constant $k$). 
In the following we will present algorithms assuming that the scenario tree is given as an input.

Given a scenario tree with transition probabilities for edges,
we can solve a natural LP-relaxation of the studied covering problem
that has variables for decisions at each node of the tree.
We study algorithms that have access to the fractional solution
and produce an integral solution for the scenario that materializes.
The quality of the produced integral solution is analyzed with respect
to the fractional solution in this particular scenario.
We obtain that for any specific scenario from the input scenario tree
the obtained integral solution is feasible with high probability,
moreover the expected cost of the integral solution is bounded with respect
to the cost of the fractional solution in this particular scenario.

We note that our approach relies on randomization, which is in contrast to
the deterministic algorithms proposed, e.g, in \cite{SwamyS05}. Notably,
the fractional solution we round may itself be expensive for some scenarios
that occur with low probability, and hence our solution for such scenarios
may be expensive in comparison to a best possible solution for this very scenario.
What we obtain is strong bounds on the expected cost, when the expectation
is with respect to two types of randomness: the randomness of scenario arrival and the randomness of the algorithm.

Unlike with the expected cost, the bound on the probability of the correctness of the solution
is independent from the fractional solution. For every scenario the fractional solution must be feasible for the scenario specific constraints, which allows us to prove that in any single scenario the probability of failure is low. Nevertheless, this again is not a standard bound on the probability over the randomness of the algorithm that all scenarios are satisfied, which would perhaps allow for a derandomization of the algorithm. The algorithms presented in this paper have substantially improved approximation ratios, but this is at the cost of the solutions being inherently randomized.

In order to simplify the presentation, we restrict our attention
to algorithms that only use the part of the fractional solution corresponding
to the already-visited nodes of the tree (on the path from the root ot the current node).
Such an algorithm can be seen as an \emph{online} algorithm that reacts 
to the piece of information revealed to the algorithm at the current node.

We can thus model a $k$-stage stochastic 
covering integer program (CIP) for our purposes as follows.
There is a \textit{hidden} covering problem ``minimize $c^T \cdot x$ subject
to $Ax \geq b$ and all variables in $x$ being non-negative integers''. 
For notational convenience for the set-cover problem, we let $n$ be the
number of rows of $A$; also, the variables in $x$ are indexed as
$x_{j,\ell}$, where $1 \leq j \leq m$ and $1 \leq \ell \leq k$.
This program, as well as a feasible \textit{fractional} solution $x^*$ 
for it, are revealed to us in $k$ stages as follows. 
In each stage $\ell$ ($1 \leq \ell \leq k$), we are given the $\ell$th-stage 
fractional values $\{x^*_{j,\ell}: ~1 \leq j \leq m\}$ of the 
variables, along with their columns in the coefficient matrix $A$,
and their coefficient in the objective function $c$. 
Given some values like this, we need to round them 
right away at stage $\ell$ using randomization if necessary, 
\textit{irrevocably}. The goal is to develop such a
rounded vector $\{y_{j,\ell}: ~1 \leq j \leq m, 1 \leq \ell \leq k\}$ that 
satisfies the
constraints $Ay \geq b$, and whose (expected) approximation ratio
$c^T \cdot y / c^T \cdot x^*$ is small. 
Our results here are summarized as follows:
\begin{theorem}
\label{thm:2stage-cip}
We obtain randomized $\lambda$-approximation algorithms for
$k$-stage stochastic CIPs for arbitrary fixed $k$, 
with values of $\lambda$ as follows. (The running time is polynomial
for any fixed $k$, and $\lambda$ is independent of $k$.)
(i) For general CIPs, with the linear system scaled so that
all entries of the matrix $A$ lie in $[0,1]$ and 
$B =  \min_{i: b_i \geq 1} b_i $, we have
$\lambda = 1 + O(\max\{(\ln n)/B, \sqrt{(\ln n)/B}\})$.
(ii) For set cover with element-degree (maximum number of given sets
containing any element of the ground set) at most $b$,
we have $\lambda = b + \epsilon$, where $\epsilon$ can be 
$N^{-C}$ with $N$ being the input-size and $C > 0$ being any constant. 
(For instance, $b = 2$ for vertex cover, where
an edge can be covered only by its two end-points.)
\end{theorem}
The ``$+ \epsilon$'' term
appears in part (ii) since the fractional solution obtained by
\cite{SwamyS05} is an $(1 + \epsilon)$--approximation to the actual LP.
We do not mention this term in part (i), by absorbing it into the
big-Oh notation. 
The two parts of this theorem are proved next. 

\subsection{A simple scheme for general CIPs}
\label{sec:cip-proof}
We use our $k$-stage model to prove Theorem~\ref{thm:2stage-cip}(i). 
We show that a simple randomized rounding approach along 
the lines of
\cite{raghavan-thompson} works here: for a suitable $\lambda \geq 1$
and independently for all $(j,\ell)$, set $x'_{j,\ell} = 
\lambda \cdot x^*_{j,\ell}$, and define the rounded value
$y_{j,\ell}$ to be $\lceil x'_{j,\ell} \rceil$ with probability 
$x'_{j,\ell} - \lfloor x'_{j,\ell} \rfloor$, and 
to be $\lfloor x'_{j,\ell} \rfloor$ with the
complementary probability of $1 - (x'_{j,\ell} - 
\lfloor x'_{j,\ell} \rfloor)$. Note that
$\E[y_{j,\ell}] = x'_{j,\ell}$. We will now show that for a suitable,
``not very large'' choice of $\lambda$, with high probability, 
all constraints will be satisfied and $c^T \cdot y$ is about
$\lambda \cdot (c^T \cdot x^*)$. 

The proof is standard, and we will illustrate it for set cover.
Note that in this case, a set of at most $n$ elements need to be
covered in the end. 
Set $\lambda = \ln n + \psi(n)$ for any arbitrarily slowly growing
function $\psi(n)$ of $n$ such that 
$\lim_{n \rightarrow \infty} \psi(n) = \infty$; run the randomized
rounding scheme described in the previous paragraph. 
Consider any finally revealed element $i$, and
let $E_i$ be the event that our rounding leaves this element
uncovered. Let $A_i$ be the family of sets in the given set-cover
instance that contain $i$; note that the fractional solution satisfies
$\sum_{j \in A_i, \ell} x^*_{j,\ell} \geq 1$. Now, if $x'_{j,\ell} \geq 1$
for some pair $(j \in A_i, ~\ell)$, then $y_{j,\ell} \geq 1$, and so, 
$i$ is guaranteed to be covered. Otherwise, 
\[ \Pr[E_i] = \prod_{j \in A_i, \ell} \Pr[y_{j,\ell} = 0]
= \prod_{j \in A_i, \ell} (1 - x'_{j,\ell}) \leq
\exp(-\sum_{j \in A_i, \ell} \lambda \cdot x^*_{j,\ell}) \leq
\exp(-\lambda) = \exp(-\psi(n))/n = o(1/n). \]
Thus, applying a union bound over the (at most $n$) finally-revealed
elements $i$, we see that $\Pr[\bigwedge_i \overline{E_i}] = 1 - o(1)$.
So, 
\[ \E[c^T \cdot y \bigm| \bigwedge_i \overline{E_i}] \leq
\frac{\E[c^T \cdot y]}{\Pr[\bigwedge_i \overline{E_i}]} =
\frac{\lambda \cdot (c^T \cdot x^*)}{\Pr[\bigwedge_i \overline{E_i}]} =
\frac{\lambda \cdot (c^T \cdot x^*)}{1 - o(1)} =
(1 + o(1)) \cdot \lambda \cdot (c^T \cdot x^*); \]
i.e., we get an $(1 + o(1)) \cdot \ln n$--approximation. Alternatively,
since $c^T \cdot y$ is a sum of independent random variables, we can
show that it is not much more than its mean, 
$\lambda \cdot (c^T \cdot x^*)$, with high probability.

The analysis for general CIPs is similar; we observe
that for any row $i$ of the constraint system, $\E[(Ay)_i] = \lambda b_i$,
use a Chernoff lower-tail bound to show that the ``bad'' event
$E_i$ that $(Ay)_i < b_i$ happens with probability noticeably smaller than
$1/n$, and apply a union bound over all $n$. Choosing 
$\lambda$ as in Theorem~\ref{thm:2stage-cip}(i) 
suffices for such an analysis; see, e.g., \cite{naor-roth}. 

\subsection{Vertex cover, and set cover with small degree}
\label{sec:vertcov-depround}
We now use a type of dependent rounding to prove 
Theorem~\ref{thm:2stage-cip}(ii). We present the case of vertex
cover ($b = 2$), and then note the small modification needed for
the case of general $b$. Note that our model  becomes the 
following for (weighted) vertex cover. There is a hidden 
undirected graph $G = (V,E)$.
The following happens for each vertex $v \in V$. We are revealed
$k$ fractional values $x^*_{v,1}, x^*_{v,2}, \ldots, x^*_{v,k}$ for $v$
one-by-one, along with the corresponding weights for $v$ (in the
objective function), $c_{v,1}, c_{v,2}, \ldots, c_{v,k}$. We aim for
a rounding $\{y_{v,\ell}\}$ that covers all edges in $E$, 
whose objective-function value $\sum_{\ell,v} c_{v,\ell} y_{v,\ell}$
is at most twice its fractional counterpart,
$\sum_{\ell,v} c_{v,\ell} x^*_{v,\ell}$.
Note that the fractional solution satisfies
\begin{equation}
\label{eqn:vertcov}
\forall (u,v) \in E, ~(\sum_{\ell=1}^k x^*_{u,\ell}) + 
(\sum_{\ell=1}^k x^*_{v,\ell}) \geq 1.
\end{equation}

Now, given a sequence $z = (z_1, z_2, \ldots, z_k)$ of values that lie
in $[0,1]$ and arrive \textit{online}, suppose we can define an efficient 
randomized procedure $\mathcal{R}$, which has the following properties:
\begin{description}
\item[(P1)] as soon as $\mathcal{R}$ gets a value $z_i$, it rounds it to
some $Z_i \in \{0,1\}$ ($\mathcal{R}$ may use the knowledge of
the values $\{z_j, Z_j: ~j < i\}$ in this process);
\item[(P2)] $\E[Z_i] \leq z_i$; and
\item[(P3)] if $\sum_i z_i \geq 1$, then at least one $Z_i$ is one
with probability one.
\end{description}
Then, we can simply apply procedure $\mathcal{R}$ independently
for each vertex $v$, to the vector $z(v)$ of \textit{scaled} values
$(\min\{2 \cdot x^*_{v,1}, 1\}, \min\{2 \cdot x^*_{v,2},1\}, \ldots, 
\min\{2 \cdot x^*_{v,k},1\})$. Property
(P2) shows that the expected value of the final solution is at most
$2 \cdot \sum_{\ell,v} c_{v,\ell} x^*_{v,\ell}$; also, since
(\ref{eqn:vertcov}) shows that for any edge $(u,v)$, at least
one of the two sums $2 \cdot \sum_{\ell} x^*_{u,\ell}$ and
$2 \cdot \sum_{\ell} x^*_{v,\ell}$ is at least $1$, property (P3) 
guarantees that each 
edge $(u,v)$ is covered with probability one. 
So, the only task is to define function $\mathcal{R}$.

For a sequence $z = (z_1, z_2, \ldots, z_k)$ arriving online, 
$\mathcal{R}$ proceeds as follows. Given $z_1$, it rounds $z_1$ to
$Z_1 = 1$ with probability $z_1$, and to
$Z_1 = 0$ with probability $1 - z_1$. 
Next, given $z_i$ for $i > 1$:

\noindent \textbf{Case I: $Z_j = 1$ for some $j < i$.} In this case,
just set $Z_i$ to $0$.

\noindent \textbf{Case II(a): $Z_j = 0$ for all $j < i$, and
$\sum_{\ell=1}^i z_{\ell} \geq 1$.} In this case,
just set $Z_i$ to $1$.

\noindent \textbf{Case II(b): $Z_j = 0$ for all $j < i$, and
$\sum_{\ell=1}^i z_{\ell} < 1$.} 
In this case, set $Z_i = 1$ with probability 
$\frac{z_i}{1 - \sum_{j=1}^{i-1} z_{j}}$, and
set $Z_i = 0$ with the complementary probability.

It is clear that property (P1) of $\mathcal{R}$ holds.
Let us next prove property (P3). 
Assume that for some $t$, $\sum_{i=1}^t z_i \geq 1$ and
$\sum_{i=1}^{t-1} z_i < 1$. It suffices to prove that
$\Pr[\exists i \leq t: Z_i = 1] = 1$. We have
\begin{eqnarray*}
\Pr[\exists i \leq t: Z_i = 1] & = &
\Pr[\exists i < t: Z_i = 1] +
\Pr[\not\exists i < t: Z_i = 1] \cdot 
\Pr[(Z_t = 1) \bigm| (Z_1 = Z_2 = \cdots Z_{t-1} = 0)] \\
& \geq &
\Pr[(Z_t = 1) \bigm| (Z_1 = Z_2 = \cdots Z_{t-1} = 0)] \\
& = & 1,
\end{eqnarray*}
from case II(a). This proves property (P3). 

We next consider property (P2), which is immediate for $i = 1$. 
If there is some $t$ such that $\sum_{i=1}^t z_i \geq 1$, take
$t$ to be the smallest such index; if there is no such $t$, define
$t = k$. The required bound of (P2), $\E[Z_i] \leq z_i$, clearly
holds for all $i > t$, since by (P3) and Case I, we have
$\E[Z_i] = 0$ for all such $i$. So suppose $i \leq t$. Note from
case II(a) that
\[ \forall j < t, ~\Pr[(Z_j = 1) \bigm| (Z_1 = Z_2 = \cdots Z_{j-1} = 0)] 
= \frac{z_j}{1 - (z_1 + z_2 + \cdots + z_{j-1})}. \]
Note from Case I that
no two $Z_j$ can both be $1$. Thus, for $1 < i \leq t$, 
\begin{eqnarray}
\Pr[Z_i = 1] & = & \Pr[(Z_1 = Z_2 = \cdots Z_{i-1} = 0) \wedge (Z_i = 1)] \nonumber \\
& = & \Pr[Z_1 = Z_2 = \cdots Z_{i-1} = 0] \cdot 
\Pr[(Z_i = 1) \bigm| (Z_1 = Z_2 = \cdots Z_{i-1} = 0)] \nonumber \\
& = & \left(\prod_{j=1}^{i-1} 
(1 - \frac{z_j}{1 - (z_1 + z_2 + \cdots + z_{j-1})})\right) \cdot
\Pr[(Z_i = 1) \bigm| (Z_1 = Z_2 = \cdots Z_{i-1} = 0)] \nonumber \\
& = & (1 - (z_1 + z_2 + \cdots + z_{i-1})) \cdot
\Pr[(Z_i = 1) \bigm| (Z_1 = Z_2 = \cdots Z_{i-1} = 0)].
\label{eqn:p2}
\end{eqnarray}
 From Cases II(a) and II(b),
$\Pr[(Z_i = 1) \bigm| (Z_1 = Z_2 = \cdots Z_{i-1} = 0)]$ is $1$
if $i = t$ \textit{and} $z_1 + z_2 + \cdots + z_t \geq 1$, and is 
$\frac{z_i}{1 - (z_1 + z_2 + \cdots + z_{i-1})}$ otherwise; in
either case, we can verify from (\ref{eqn:p2}) that
$\Pr[Z_i = 1] \leq z_i$, proving (P2).

Similarly, for $k$-stage set cover with each element of the universe 
contained in at most $b$ of the given sets, we construct 
$z'(v) = (\min\{b \cdot x^*_{v,1}, 1\}, 
\min\{b \cdot x^*_{v,2}, 1\}, 
\ldots, \min\{b \cdot x^*_{v,k}, 1\})$ and apply $\mathcal{R}$. By the same
analysis as above, all elements are covered with probability $1$,
and the expected objective function value is at most 
$b \cdot \sum_{\ell,v} c_{v,\ell} x^*_{v,\ell}$.

\noindent \textbf{Tail bounds:} It is also easy to show using
\cite{ps:edge-col} that
in addition to its expectation being at most $b$ times the fractional
value, $c^T \cdot y$ has a
Chernoff-type upper bound on deviations above its mean. 

%

\section{Facility Location Problems}
\label{sec:facloc-exp}
We consider three variants of facility location in this section
(i.e., the \emph{standard}, the \emph{expected per-scenario guarantee}, and the \emph{strict per-scenario guarantee} models, see \S~\ref{sec:intro} for definitions).
We start with a randomized primal-dual $2.369$-approximation algorithm in \S~\ref{sec:ufl-2stage-expmin}. 
Then, in \S~\ref{sec:LP-rounding} we give an LP-rounding algorithm (based on a dual bound on maximal connection cost) 
with a \emph{bifactor approximation guarantee} $(2.4061, 1.2707)$, i.e., one that delivers solutions with cost at most $1.2707$ times the fractional connection cost
plus $2.4061$ times the fractional facility opening cost. This algorithm is trivially a $2.4061$-approximation algorithm.
Next, in \S~\ref{sec:per-scenario}, we give a purely primal LP-rounding algorithm which has a $2.4975$-approximation guarantee
in the \emph{expected per-scenario} sense.

In \S~\ref{sec:ufl-combined} we further exploit the asymmetry of the analysis of the first LP-rounding algorithm and combine
it with the algorithm from \S~\ref{sec:ufl-2stage-expmin}. As a result we obtain an improved approximation guarantee
of $2.2975$ for the standard setting of $2$-stage stochastic uncapacitated facility location, where the goal is to optimize
the expected total cost (across scenarios). 

The second LP-rounding algorithm not only has the advantage of providing solutions where the expected cost in each scenario is bounded: 
it also can be applied in the black-box model. Finally, in \S~\ref{sec:strict-per-scenario} we note that one can obtain
an algorithm that deterministically obeys certain \emph{a priori} per-scenario budget constraints by splitting a single two-stage instance 
into two single stage instances.

We will consider just the case 
of $0-1$ demands. As usual, our algorithms directly generalize to the 
case of arbitrary demands with no loss in approximation guarantee. 

\subsection{General setting}
Let the set of facilities be $\mathcal{F}$, and the set of all
possible clients be $\mathcal{D}$. From the
results of \cite{ShmoysS04,CharikarCP05,SwamyS05,SoZY06}, we may assume
that we are given:
(i) $m$ scenarios (indexed by $A$), each being a subset of $\mathcal{D}$, and
(ii) an $(1 + \epsilon)$-approximate solution $(x,y)$ to the following
standard LP relaxation of the problem (as in Theorem~\ref{thm:2stage-cip},
$\epsilon$ can be made inverse-polynomially small, and will henceforth be
ignored):
\[ \mbox{minimize} ~
\sum_{i \in \mathcal{F}} f_i^I y_i +
\sum_A p_A (\sum_i f_i^A y_{A,i} + \sum_{j \in A} \sum_i c_{ij} x_{A,ij})~
\mbox{subject to} \]
\begin{eqnarray}
\sum_i x_{A,ij} & \geq & 1 ~~\forall A ~\forall j \in A; 
\label{eqn:facloc-assign} \\
x_{A,ij} & \leq & y_i + y_{A,i} ~~\forall i ~\forall A ~\forall j \in A; 
\label{eqn:facloc-xleqy} \\
x_{A,ij}, y_i, y_{A,i} & \geq & 0 ~~\forall i ~\forall A ~\forall j \in A.
\nonumber
\end{eqnarray}
Here, $f_i^I$ and $f_i^A$ are the costs of opening facility $i$ in stage
I and in stage-II scenario $A$, respectively; $c_{ij}$ is the cost of
connecting client $j$ to $i$. Each given scenario $A$ materializes
with probability $p_A$. Variables $y_i$ and $y_{A,i}$ are the
extents to which facility $i$ is opened in stage I and in stage-II 
scenario $A$, respectively; $x_{A,ij}$ is the extent to which
$j$ is connected to $i$ in scenario $A$.

For all $i$, $A$, and $j \in A$ such that $x_{A,ij} > 0$, write
$x_{A,ij} = x_{A,ij}^{(1)} + x_{A,ij}^{(2)}$, where
\begin{equation}
\label{eqn:x-decomp}
x_{A,ij}^{(1)} = x_{A,ij} \cdot \frac{y_i}{y_i + y_{A,i}} 
\mbox{ and } 
x_{A,ij}^{(2)} = x_{A,ij} \cdot \frac{y_{A,i}}{y_i + y_{A,i}}.
\end{equation}
Extending this definition, if $j \in A$ and $x_{A,ij} = 0$, we
define $x_{A,ij}^{(1)} = x_{A,ij}^{(2)} = 0$. 
Note from (\ref{eqn:facloc-xleqy}) that
$x_{A,ij}^{(1)} \leq y_i$ and $x_{A,ij}^{(2)} \leq y_{A,i}$.

The idea is, as in \cite{ShmoysS04}, to satisfy some client-scenario 
pairs $(j,A)$ in Stage I; the rest will be handled in Stage II.
This set of pairs are chosen based on the values of $\sum_i x_{A,ij}^{(1)}$.
Our contribution is that we propose alternatives to the direct use of ``deterministic
thresholding'', which is to choose such Stage-I pairs as in \cite{ShmoysS04}
and use existing algorithms for the obtained subproblems.

In the first primal-dual algorithm we will employ a carefully-chosen \textit{randomized} thresholding.
As we will see, this randomized scheme also fits well with a basic 
\textit{asymmetry} in many known facility-location algorithms (in our case,
the ones in \cite{JainMMSV02,JainMS02}). 
In the LP-rounding algorithms presented later, we use a deterministic threshold: 
however, instead of individually solving the obtained subproblems, we 
perform a single rounding process. The selection of Stage I client-scenario pairs
is only used to guide the clustering of facilities, but we still allow a single client
to connect to a facility opened in either of the stages.

\subsection{Minimizing expected cost: a primal-dual algorithm}
\label{sec:ufl-2stage-expmin}
We now develop a $2.369$--approximation algorithm for minimizing the expected
total cost.  

Let $\alpha \in (0,1/2)$ be a constant
that will be chosen later. Pick a single random real $Z$ using the following
distribution that is a mixture of continuous and discrete: 
\begin{itemize}
\item with probability $\alpha/(1 - \alpha)$, let $Z := 1/2$;
\item with the complementary probability of $(1 - 2\alpha)/(1 - \alpha)$, 
let $Z$ be a random real chosen from the uniform distribution on
$[\alpha, 1 - \alpha]$.
\end{itemize}

The rounding for Stage I is as follows. For any pair $(j,A)$ with 
$j \in A$, define $r_{A,j}^{(1)}$ (the
extent to which $(j,A)$ is satisfied in Stage I) to be 
$\sum_i x_{A,ij}^{(1)}$; $(j,A)$ is
declared \textit{selected} iff 
\begin{equation}
\label{eqn:stage1-select}
Z \leq r_{A,j}^{(1)}.
\end{equation}
For the Stage I decisions, construct a facility-location instance
$\mathcal{I}$ with each selected pair $(j,A)$ having 
demand $p_A$ and each facility
$i$ having cost $f_i^I$, and solve it
using the approximation algorithm of \cite{JainMMSV02,JainMS02},
which is described in \cite{MahdianYZ05} and called the 
\textit{JMS Algorithm} in \cite{MahdianYZ05}. 
In Stage II, we round separately for each scenario $A$ as follows.
Construct a facility-location instance $\mathcal{I}_A$
with a unit-demand client for
each $j \in A$ such that $(j,A)$ was \textit{not} selected in Stage I;
each facility $i$ has cost $f_i^A$. Again use the JMS algorithm as
described in \cite{MahdianYZ05} to get an approximately optimal solution
for $\mathcal{I}_A$. 

\smallskip \noindent \textbf{Analysis:}
It is clear that in every scenario $A$, we satisfy all of its demands.
To analyze the expected cost of this solution (with the only randomness
being in the choice of $Z$), we start by constructing feasible fractional
solutions for the facility-location instances $\mathcal{I}$
and $\mathcal{I}_A$ (for all $A$). Condition on a fixed value for
$Z$. Let us first construct a feasible fractional solution $(\hat{x},
\hat{y})$ for the stage-I instance $\mathcal{I}$:
$\hat{y}_i = y_{i} / Z$ for all $i$, and
$\hat{x}_{A,ij} = x_{A,ij}^{(1)} / r_{A,j}^{(1)}$ for all
selected $(j,A)$ and all $i$. This is feasible since
$r_{A,j}^{(1)} \geq Z$. Thus, letting $S_{j,A}$ be the indicator
variable for $(j,A)$ being selected (which is a function of $Z$)
and recalling that each selected $(j,A)$ has demand $p_A$ in
$\mathcal{I}$, 
the total ``facility cost'' and ``connection cost'' of $(\hat{x}, \hat{y})$
are
\begin{equation}
\label{eqn:stageI-costs}
\sum_i \frac{y_{i} f_i^{I}}{Z}
\mbox{ and }
\sum_{j,A} p_A \cdot \frac{S_{j,A}}{r_{A,j}^{(1)}} \cdot \sum_i 
c_{ij} x_{A,ij}^{(1)},
\end{equation}
respectively.
Next consider any scenario $A$, and let us construct 
a feasible fractional solution $(x', y')$ for $\mathcal{I_A}$.
Define $r_{A,j}^{(2)} = \sum_i x_{A,ij}^{(2)}$. 
We may assume w.l.o.g.\ that equality
holds in (\ref{eqn:facloc-assign}); so, 
$r_{A,j}^{(2)} = 1 - r_{A,j}^{(1)}$. Thus, a
necessary condition for $(j,A)$ to \textit{not} be selected
in Stage I is
\begin{equation}
\label{eqn:stage2-select}
(1 - Z) \leq r_{A,j}^{(2)}.
\end{equation}
This is analogous to (\ref{eqn:stage1-select}), with $Z$ being replaced
by $1 - Z$. Thus, we can argue similarly as we did for $(\hat{x}, \hat{y})$ 
that
$y'_i = y_{A,i} / (1 - Z)$, 
$x'_{A,ij} = x_{A,ij}^{(2)} / r_{A,j}^{(2)}$ for all $(j,A)$ not selected
in Stage I, is a feasible fractional solution for $\mathcal{I_A}$. Since
all demands here are one, 
the total facility cost and connection cost of $(x', y')$
are
\begin{equation}
\label{eqn:stageII-costs}
\sum_i \frac{y_{A,i} f_i^A}{1 - Z}
\mbox{ and }
\sum_{j,A} \frac{1 - S_{j,A}}{r_{A,j}^{(2)}} \cdot \sum_i c_{ij} x_{A,ij}^{(2)}
\end{equation}
respectively.

Now, the key ``asymmetry'' property of the JMS algorithm is, as proven in
\cite{MahdianYZ05}, that it is a \textit{$(1.11, 1.78)$-bifactor 
approximation algorithm}: given an instance for which there is a fractional
solution with facility cost $F$ and connection cost $C$, it produces an
integral solution of cost at most $1.11 F + 1.78 C$. Therefore, 
from (\ref{eqn:stageI-costs}) and (\ref{eqn:stageII-costs}), and weighting
the latter by $p_A$, we see that given $Z$, the total final cost is
at most
\[ 1.11 \cdot \left[\sum_i \left(\frac{y_{i}f_i^I}{Z} + \sum_A p_A \cdot \frac{y_{A,i}f_i^A}{1 - Z}\right)\right] +
1.78 \cdot 
\sum_{j,A} p_A \cdot \left[\left(\frac{S_{j,A}}{r_{A,j}^{(1)}} \cdot \sum_i 
c_{ij} x_{A,ij}^{(1)}\right) + 
\left(\frac{1 - S_{j,A}}{r_{A,j}^{(2)}} \cdot \sum_i c_{ij} x_{A,ij}^{(2)}\right)\right]; \]
so, the expected final cost is at most
\begin{eqnarray}
& 1.11 \cdot \left[\sum_i (y_{i} \cdot \E[1/Z] + 
\sum_A p_A y_{A,i} \cdot \E[1/(1 - Z)])\right] + & \nonumber \\
& 1.78 \cdot 
\sum_{j,A} p_A \cdot \left[\left(\frac{\E[S_{j,A}]}{r_{A,j}^{(1)}} \cdot \sum_i 
c_{ij} x_{A,ij}^{(1)}\right) + 
\left(\frac{\E[1 - S_{j,A}]}{r_{A,j}^{(2)}} \cdot \sum_i c_{ij} x_{A,ij}^{(2)}\right)\right]. & \label{eqn:facloc-cost}
\end{eqnarray}

Note that $Z$ and $1 - Z$ have identical distributions. So,
\begin{equation}
\label{eqn:e1z}
\E[1/(1 - Z)] = \E[1/Z] = 
(\alpha/(1 - \alpha)) \cdot 2 + 
((1 - 2\alpha)/(1 - \alpha)) \cdot 
\frac{1}{1 - 2\alpha} \cdot \int_{z = \alpha}^{1 - \alpha} dz / z
= \frac{2 \alpha + \ln((1 - \alpha) / \alpha)}{1 - \alpha}.
\end{equation}
Let us next bound $\E[S_{j,A}]$. Recall 
(\ref{eqn:stage1-select}), and let $r$ denote $r_{A,j}^{(1)}$. If
$r < \alpha$, then $S_{j,A} = 0$; if $r \geq 1 - \alpha$, then
$S_{j,A} = 1$. Next suppose $\alpha \leq r < 1/2$. Then
$S_{j,A}$ can hold only if we chose to pick $Z$ at random from
$[\alpha, 1 - \alpha]$, and got $Z \leq r$; this happens with probability
$((1 - 2\alpha)/(1 - \alpha)) \cdot (r - \alpha) / (1 - 2\alpha) 
= (r - \alpha)/(1 - \alpha) \leq r/(1 - \alpha)$. Finally, if
$1/2 \leq r < (1 - \alpha)$,
\[ \E[S_{j,A}] = \alpha/(1 - \alpha) + ((1 - 2\alpha)/(1 - \alpha)) \cdot 
(r - \alpha) / (1 - 2\alpha) = r/(1 - \alpha). \]
Thus, in all cases we saw here, 
\begin{equation}
\label{eqn:esja}
\E[S_{j,A}] \leq r_{A,j}^{(1)} / (1 - \alpha). 
\end{equation}
Similarly, recalling (\ref{eqn:stage2-select}) and the fact that
$Z$ and $1 - Z$ have identical distributions, we get
\begin{equation}
\label{eqn:e1-sja}
\E[1 - S_{j,A}] \leq r_{A,j}^{(2)} / (1 - \alpha). 
\end{equation}

Plugging (\ref{eqn:e1z}), (\ref{eqn:esja}), and (\ref{eqn:e1-sja}) into
(\ref{eqn:facloc-cost}) and using the fact that 
$x_{A,ij} = x_{A,ij}^{(1)} + x_{A,ij}^{(2)}$, we see that
our expected approximation ratio is
\[ \max\left\{\frac{1.78}{1 - \alpha}, ~\frac{1.11(2 \alpha + \ln((1 - \alpha) / \alpha))}{1 - \alpha} \right\}. \]
Thus, a good choice of $\alpha$ is $0.2485$,
leading to an expected approximation ratio less than $2.369$.


\subsection{Minimizing expected cost: an LP-rounding algorithm}
\label{sec:LP-rounding}
Consider the following dual formulation of the $2$-stage stochastic facility location problem: 
\begin{eqnarray*}
 \mbox{maximize} \sum_{A} p_A (\sum_{j \in A} v_{j,A}) && 
~\mbox{subject to:}\\
   v_{j,A} - c_{ij} &\leq& w_{ij,A} ~~\forall i ~\forall A ~\forall j \in A\\
   \sum_{A} p_A (\sum_{j \in A} w_{ij,A}) &\leq& f_i^I ~~~~~\forall i ~\forall A ~\forall j \in A\\
   \sum_{j \in A} w_{ij,A} &\leq& f_i^A ~~~~~\forall i ~\forall A ~\forall j \in A\\
  w_{ij,A},v_{j,A} &\geq& 0 ~~~~~~~\forall i ~\forall A ~\forall j \in A.
\end{eqnarray*}

Let $(x^*,y^*)$ and $(v^*,w^*)$ be optimal solutions to the primal
and the dual programs, respectively. Note that by complementary 
slackness, we have
$c_{ij} \leq v_{j,A}$ if $x_{A,ij} > 0$. 

\paragraph{Algorithm.}

We now describe a randomized LP-rounding algorithm that transforms the fractional 
solution $(x^*,y^*)$ into an integral solution $(\hat{x},\hat{y})$ with bounded expected cost.
The expectation is over the random choices of the algorithm, but not over the
random choice of the scenario. Note that we need to decide the first stage entries of $\hat{y}$
not knowing $A$. W.l.o.g., we assume that no facility is fractionally opened in $(x^*,y^*)$ 
in both stages, i.e, for all $i$ we have $y^*_i = 0$ or for all $A$ $y^*_{A,i} = 0$.
To obtain this property it suffices to have two identical copies of each facility,
one for Stage I and one for Stage II.

Define $x_{A,ij}^{(1)}$, $x_{A,ij}^{(2)}$, and $r_{A,j}^{(1)}$ as before
and select all the $(j,A)$ pairs with $r_{A,j}^{(1)} \geq \frac{1}{2}$ (note that it is the standard deterministic tresholding selection method).
We will call such selected $(j,A)$ pairs \emph{first stage clustered}
and the remaining $(j,A)$ pairs \emph{second stage clustered}. Let $S$ denote the set of firs stage clustered $(j,A)$ pairs.

We will now scale the fractional solution by a factor of $2$.
Define $\overline{x}_{A,ij}^{(1)} = 2 \cdot x_{A,ij}^{(1)} $, $\overline{x}_{A,ij}^{(2)} = 2 \cdot x_{A,ij}^{(2)} $, $\overline{y}_i =  2 \cdot y^*_i $,
and $\overline{y}_{A,i} =  2 \cdot y^*_{A,i} $. Note that the scaled fractional solution
$(\overline{x},\overline{y})$ can have facilities with fractional opening of more than $1$.
For simplicity of the analysis, we do not round these facility-opening values to $1$, but rather split such facilities.
More precisely, we split each facility $i$ with fractional opening $\overline{y}_i > \overline{x}_{A,ij}^{(1)} > 0$ (or $\overline{y}_{A,i} > \overline{x}_{A,ij}^{(2)} > 0$) 
for some $(A, j)$ into $i'$ and $i''$, such that 
 $\overline{y}_{i'} = \overline{x}_{A,ij}^{(1)}$ and $\overline{y}_{i''} = \overline{y}_i - \overline{x}_{A,ij}^{(1)}$.
We also split facilities whose fractional opening exceeds one.
By splitting facilities we create another instance of the problem together with a fractional sollution, 
then we solve this modified instance and interpret the solution as a solution to the original problem in the natural way.
The technique of splitting facilities is precisely described in~\cite{Sviridenko02}.


Since we can split facilities, for each $(j,A) \in S$ we can assume that
there exists a subset of facilities $F_{(j,A)} \subseteq \F$, 
such that $\sum_{i \in  F_{(j,A)} } \overline{x}^{(1)}_{A,ij} = 1$,
and for each $i \in F_{(j,A)}$ we have $\overline{x}^{(1)}_{A,ij} = \overline{y}_i$.
Also for each $(j,A) \notin S$ we can assume that
there exists a subset of facilities $F_{(j,A)} \subseteq \F$,
such that $\sum_{i \in  F_{(j,A)} } \overline{x}^{(2)}_{A,ij} = 1$,
and for each $i \in F_{(j,A)}$ we have $\overline{x}^{(2)}_{A,ij} = \overline{y}_{A,i}$.
Let $R_{(j,A)} = \max_{i \in F_{(j,A)}} c_{ij}$ be the maximal distance from $j$ to an $i \in F_{(j,A)}$.
Recall that, by complementary slackness, we have $R_{(j,A)} \leq v_{j,A}^{*}$. 

The algorithm opens facilities randomly in each of the stages
with the probability of opening facility $i$ equal to $\overline{y}_i$
in Stage I, and $\overline{y}_{A,i}$ in Stage II of scenario $A$.
Some facilities are grouped in disjoint \emph{clusters} in order to correlate 
the opening of facilities from a single cluster.
The clusters are formed in each stage by the following procedures.
Let all facilities be initially unclustered.

In Stage I, consider all client-scenario pairs $(j,A) \in S$
in the order of non-decreasing values $R_{(j,A)}$.
If the set of facilities $F_{(j,A)}$ contains no facility
from the previously formed clusters, 
then form a new cluster containing facilities from $F_{(j,A)}$,
otherwise do nothing.
Recall that the total fractional opening of facilities in each cluster
equals $1$. Open exactly one facility in each cluster. Choose the facility randomly with the probability
of opening facility $i$ equal to the fractional opening $\overline{y}_i$.
For each unclustered facility $i$ open it independently with 
probability $\overline{y}_i$.

In Stage II of scenario $A$, consider all clients $j$ such that $(j,A) \notin S$
in the order of non-decreasing values $R_{(j,A)}$.
If the set of facilities $F_{(j,A)}$ contains no facility
from the previously formed clusters, 
then form a new cluster containing facilities from $F_{(j,A)}$,
otherwise do nothing.
Recall that the total fractional opening of facilities in each cluster
equals $1$. Open exactly one facility in each cluster. Choose the facility randomly with the probability
of opening facility $i$ equal to the fractional opening $\overline{y}_{A,i}$.
For each unclustered facility $i$ open it independently with 
probability $\overline{y}_{A,i}$.

Finally, at the end of Stage II of scenario $A$, connect each client $i \in A$ to the 
closest open facility (this can be a facility open in Stage I or in Stage II).

\paragraph{Expected-distance Lemma.} Before we proceed to bounding the expected cost of the 
solution obtained by the above algorithm, let us first bound the expected distance to 
an open facility from the set of facilities that fractionally service
client in the solution $(x^*,y^*)$. Denote by $C_{(j,A)} = \sum_{i \in \mathcal{F}}c_{ij}x^*_{A,ij}$
the fractional connection cost of client $j$ in scenario $A$ in the fractional solution $(x^*,y^*)$.
For the purpose of the next argument we fix a client-scenario pair $(j,A)$,
and slightly change the notation and let vector $\overline{y}$ encode fractional opening of both the first stage facilities 
and the second stage facilities of scenario $A$, all of them referred to with a singe index $i$.
A version of the following argument was used in the analysis of most of the previous
LP-rounding approximation algorithms for facility location problems, see, e.g.,
\cite{ByrkaAardal10}.  

\begin{lemma} \label{exp_distance_lemma}
 Let $y \in \{0,1\}^{|\mathcal{F}|}$ be a random binary vector encoding the facilities opened by the algorithm, 
 let $F'\subseteq \mathcal{F}$ be the set of facilities fractionally servicing client $j$ in scenario $A$, then:
\[
 E\left[\min_{i\in F', y_i=1}c_{ij} \ | \ \sum_{i\in F'} y_i \geq 1\right] \leq C_{(j,A)}. 
\]
\end{lemma}
\begin{proof}
 Let $C_1$, \ldots $C_k$ be clusters intersecting $F'$. They partition $F'$ into $k+1$ disjoint sets of facilities: $F_0$ not intersecting any cluster, 
 and $F_i$ intersecting $C_i$ for $i= 1 \ldots k$. Note that opening of facilities in different sets $F_i$ is independent. Let $c_i$ be the average distance between $j$ and facilities in $F_i$. Observe that we may ignore the differences between facilities within sets $F_i$, $i \geq 1$,  and treat them as single facilities at distance $c_i$ with fractional opening equal the total opening of facilities in the corresponding set, because these are the expected distance and the probability that a facility will be opened in $F_i$. It remains to observe that the lemma obviously holds for the remaining case where $F'$ only contains facilities whose opening variables $y$ are rounded independently (preserving marginals).
\end{proof}

\paragraph{Analysis.}
Consider the solution $(\hat{x},\hat{y})$ constructed by our LP-rounding algorithm.
We fix scenario $A$ and bound the expectation of 
$COST(A) = \sum_{i \in \F} (f_i^I \hat{y}_{i} + f_i^A \hat{y}_{A,i}) + \sum_{j\in A}\sum_{i\in \F} c_{ij}\hat{x}_{A,ij}$.
Define $C_A = \sum_{j \in A} C_{(j,A)} = \sum_{j \in A} \sum_{i \in \mathcal{F}}c_{ij}x^*_{A,ij}$,
$F_A = \sum_{i \in \F} (f_i^I y^*_{i} + f_i^A y^*_{A,i})$, $V_A = \sum_{j \in A} v^*_{j,A}$.

\begin{lemma} \label{primal_dual_lemma}
  $E[COST(A)] \leq e^{-2} \cdot 3 \cdot V_A + (1-e^{-2}) \cdot C_A + 2 \cdot F_A$ in each scenario $A$.
\end{lemma}

\begin{proof}
Since the probability of opening a facility is
equal to its fractional opening in $(\overline{x},\overline{y})$, 
the expected facility-opening cost of $(\hat{x},\hat{y})$ 
equals facility-opening cost of $(\overline{x},\overline{y})$, which
is exactly twice the facility-opening cost of $(x^*,y^*)$.

Fix a client $j \in A$. 
The total (from both stages) fractional opening in $\overline{y}$ of facilities serving $j$ in 
$(\overline{x},\overline{y})$ is exactly $2$, hence the probability that at least one of these 
facilities is open in $(\hat{x},\hat{y})$ is at least $1-e^{-2}$.
Observe that, on the condition that at least one such facility is open,
by Lemma~\ref{exp_distance_lemma}, 
the expected distance to the closest of the open facilities is at most $C_{(j,A)}$.

With probability at most $e^{-2}$, none of the facilities fractionally serving $j$
in $(\overline{x},\overline{y})$ is open. In such a case we need to find a different
facility to serve $j$. We will now prove that for each client $j \in A$ there exists a facility
$i$ which is open in $(\hat{x},\hat{y})$, such that $c_{ij} \leq 3 \cdot v_{j,A}$.

Assume $(j,A) \in S$ (for $(j,A) \notin S$ the argument is analogous). 
If $F_{(j,A)}$ is a cluster, then at least one $i \in F_{(j,A)}$ is open
and $c_{ij} \leq v_{j,A}$.
Suppose $F_{(j,A)}$ is not a cluster, then by the construction of clusters,
it intersects a cluster $F_{(j',A')}$ with $R_{(j',A')} \leq R_{(j,A)} \leq v_{j,A}$.
Let $i$ be the facility opened in cluster $F_{(j',A')}$ and let $i' \in F_{(j',A')} \cap F_{(j,A)}$.
Since $i'$ is in $F_{(j,A)}$, $c_{i'j} \leq R_{(j,A)}$. Since both $i$ and $i'$ are in $F_{(j',A')}$,
both $c_{ij'} \leq R_{(j',A')}$ and $c_{i'j'} \leq R_{(j',A')}$.
Hence, by the 
triangle inequality, $c_{ij} \leq R_{(j,A)} + 2 \cdot R_{(j',A')} \leq 3 \cdot R_{(j,A)} \leq 3 \cdot v_{j,A}$.

Thus, the expected cost of the solution in scenario $A$ is:
\begin{eqnarray*} 
 E[COST(A)] & \leq & e^{-2} \cdot 3 \cdot \sum_{j \in A} v_{j,A}^* + 
 (1-e^{-2})  (\sum_{j\in A}\sum_{i\in \F} c_{ij}x^*_{A,ij}) + 
 2 \cdot (\sum_{i \in \F} (f_i^I y^*_{i} + f_i^A y^*_{A,i})),
\end{eqnarray*}
\rev{that} is \rev{exactly} 
$e^{-2} \cdot 3 \cdot V_A + (1-e^{-2}) \cdot C_A + 2 \cdot F_A$.
\end{proof}

Define $F^* = \sum_{i \in \mathcal{F}} f_i^I y_i +
\sum_A p_A (\sum_i f_i^A y_{A,i})$
and $C^* = \sum_A p_A (\sum_{j \in A} \sum_i c_{ij} x_{A,ij})$.
Note that we have $F^* = \sum_A p_A F_A$, $C^* = \sum_A p_A C_A$, and
$F^*+C^* = \sum_A p_A V_A$.
Summing up the expected cost over scenarios we obtain the following
estimate on the general expected cost, where the expectation is both 
over the choice of the scenario and over the random choices of our algorithm.

\begin{corollary}
 $E[COST(\hat{x},\hat{y})] \leq 2.4061 \cdot F^* + 1.2707 \cdot C^*$.
\end{corollary}

\begin{proof}
$E[COST(\hat{x},\hat{y})]$ is $\sum_A p_A E[COST(A)]$, which is at most
\begin{eqnarray*}
\sum_A p_A\left( e^{-2} \cdot 3 \cdot V_A + (1-e^{-2}) \cdot C_A + 2 \cdot F_A \right) & = & (1-e^{-2}) \cdot C^* + 2 \cdot F^* + 3e^{-2} (\sum_A p_A V_A)\\
 & = & (1-e^{-2}) \cdot C^* + 2 \cdot F^* + 3e^{-2} (F^*+C^*)\\
 & = & (2 + e^{-2} \cdot 3) F^* + (1 + e^{-2} \cdot 2) C^* \\
 & \leq & 2.4061 \cdot F^* + 1.2707 \cdot C^*.
\end{eqnarray*}
\end{proof}

\subsection{Two algorithms combined}
\label{sec:ufl-combined}

We will now combine the algorithms from \S~\ref{sec:ufl-2stage-expmin} and \S~\ref{sec:LP-rounding}
to obtain an improved approximation guarantee for the problem
of minimizing the expected cost over the choice of the scenario.

To this end we will analyze the cost of the computed solution with respect to the 
the facility opening cost $F$ and the connection cost $C$ of an optimal solution to the problem.

In \S~\ref{sec:LP-rounding} we gave a rounding procedure whose cost is bounded in terms of
the cost of the initial fractional solution. It was shown that the algorithm returns a soulution of cost 
at most $2.4061$ times the facility opening cost plus $1.2707$ timest the connection cost of the fractional solution.
Observe, that it suffices to apply this procedure to a fractional solution obtained by solving a properly scaled LP to get an algorithm with a corresponding bifactor approximation guarantee. In particular, if we scale the 
facility opening costs by $2.4061$ and the connection costs $c_{ij}$ by $1.2707$ before solving the LP,
and then round the obtained fractional solution with the described procedure, we obtain a solution of cost at most 
$2.4061 F + 1.2707 C$. We will call this algorithm ALG1.

Now consider the algorithm discussed in \S~\ref{sec:ufl-2stage-expmin}. For the choice of a parameter
$\alpha = 0.2485$ it was shown to be a $2.369$-approximation algorithm. We will now consider 
a different choice, namely $\alpha = 0.37$. It is easy to observe that it results in an algorithm computing solutions of cost at most $2.24152 F + 2.8254 C$. We will call this algorithm ALG2.

Consider the algorithm ALG3, which tosses a coin that comes heads with probability $p=0.3396$.
If the coin comes heads, then ALG1 is executed; if it comes tails ALG2 is used.
The expected cost of the solution may be estimated as: 
$(p \cdot 2.4061 + (1-p)\cdot 2.24152) F + (p \cdot 1.2707 + (1-p)\cdot 2.8254) C \leq 2.2975 (F + C)$.
Therefore, ALG3 is a $2.2975$-approximation algorithm for the $2$-stage stochastic facility location problem.
Note that the initial coin tossing in ALG3 may be derandomized by running both ALG1 and ALG2
and taking the better of the solutions.

\subsection{Facility location with per-scenario bounds}
\label{sec:per-scenario}
Consider again the $2$-stage facility location problem, and a corresponding
optimal fractional solution. We now describe a randomized rounding
scheme so that for each scenario $A$, its expected final (rounded)
cost is at most $2.4957$ times its fractional counterpart $Val_A = \sum_{i \in \F} (f_i^I y^*_i + f_i^A y^*_{A,i}) + \sum_{j\in A}\sum_{i\in \F} c_{ij}x^*_{A,ij}$,
improving on the $3.225 \cdot Val_A$ bound of \cite{Swamy04} and the $3.095 \cdot Val_A$
bound from an earlier version of this paper \cite{Srinivasan07}.


Let us split the value of the fractional solution $Val_A$
into the fractional connection cost $C_A = \sum_{j\in A}\sum_{i\in \F} c_{ij}x^*_{A,ij}$ 
and fractional facility-opening cost $F_A = \sum_{i \in \F} (f_i^I y^*_i + f_i^A y^*_{A,i})$ 

Before we proceed with the per-scenario algorithm, 
let us first note that it is not possible to directly use the analysis from
the previous setting in the per-scenario model. This is because the dual costs $V_A$
do not need to be equal $Val_A = F_A + C_A$ in each scenario $A$. It is possible, for instance,
that the fractional opening of a facility in the first stage is entirely paid 
from the dual budget of a single scenario, despite the fact that clients not active
in this scenario benefit from the facility being open. This can be observed, for instance, in the following
simple example. 

Consider two clients $c^1$ and $c^2$, and two facilities $f^1$ and $f^2$.
All client facility distances are $1$, except $c_{1,2}=dist(c^1,f^2)=3$.
Scenarios are: $A^1 = \{c^1\}$ and $A^2=\{c^2\}$, and they occur with probability $1/2$ each.
The facility-opening costs are: $f_1^I=2$, $f_2^I=\epsilon$, $f_1^A=f_2^A=4$ for both scenarios $A$. 
It is easy to see that the only optimal fractional solution is integral
and it opens facility $f^1$ in the first stage, and opens no more facilities
in the second stage. Therefore, $Val(A^1)=Val(A^2)=3$.
However, in the dual problem, client $c^2$ has an advantage over $c^1$ in the access 
to the cheaper facility $f^2$, and therefore in no optimal dual solution client $c_2$ will pay
more than $\epsilon$ for the opening of facility $f^1$. In consequence, most of the cost of opening $f^1$
is paid by the dual budget of scenario $A^1$. Therefore, the dual budget $V_{A^1}$ 
is strictly greater than the primal bound $Val_{A^1}$ which we use as an estimate
of the cost of the optimal solution in scenario $A^1$.

Bearing the above example in mind, we construct an LP-rounding algorithm
that does not rely on the dual bound on the length of the created connections.
We use a primal bound, which is obtained by scaling the opening variables a little more
and using just a subset of fractionally connected facilities for each client
in the process of creating clusters. Such a simple filtering technique,
whose origins can be found in the work of Lin and Vitter~\cite{DBLP:conf/stoc/LinV92}, provides slightly weaker
but entirely primal, per-scenario bounds.

\paragraph{Algorithm.}

As before, we describe a randomized LP-rounding algorithm that transforms the fractional 
solution $(x^*,y^*)$ into an integral solution $(\hat{x},\hat{y})$; 
the expectation of the cost that we compute is over the 
random choices of the algorithm, but not over the
random choice of the scenario. 
Again, we assume that no facility is fractionally opened in $(x^*,y^*)$ 
in both stages.

Again, define $x_{A,ij}^{(1)}$, and $x_{A,ij}^{(2)}$ as before. However, the set of \emph{first stage clustered} pairs $(j,A)$
will now be determined differently.

We will now scale the fractional solution by a factor of $\gamma > 2$.
Define $\overline{x}_{A,ij} =\gamma \cdot x^*_{A,ij}$,  $\overline{x}_{A,ij}^{(1)} = \gamma \cdot x_{A,ij}^{(1)} $, $\overline{x}_{A,ij}^{(2)} = \gamma \cdot x_{A,ij}^{(2)} $, 
$\overline{y}_i =  \gamma \cdot y^*_i $, and $\overline{y}_{A,i} =  \gamma \cdot y^*_{A,i} $. Note that the scaled fractional solution
$(\overline{x},\overline{y})$ can have facilities with fractional opening of more than $1$.
For simplicity of the analysis, we do not round these facility-opening values to $1$, but rather split such facilities.
More precisely, we split each facility $i$ with fractional opening $\overline{y}_i > \overline{x}_{A,ij}^{(1)} > 0$ (or $\overline{y}_{A,i} > \overline{x}_{A,ij}^{(2)} > 0$) 
for some $(A, j)$ into $i'$ and $i''$, such that 
 $\overline{y}_{i'} = \overline{x}_{A,ij}^{(1)}$ and $\overline{y}_{i''} = \overline{y}_i - \overline{x}_{A,ij}^{(1)}$.
We also split facilities whose fractional opening exceeds one.
%
%

Define \\ 
\[
F_{(j,A)}^I = \left\{ 
\begin{array}{l l}
    \rev{\mathsf{argmin}}_{F' \subseteq \F : \sum_{i \in F'} \overline{x}^{(1)}_{A,ij} \geq 1 } \rev{\mathsf{max}}_{i \in F'} c_{ij} &  
    \mbox{if $\sum_{i \in \F} \overline{x}^{(1)}_{A,ij} \geq 1$}\\
  \emptyset &  \mbox{if $\sum_{i \in \F} \overline{x}^{(1)}_{A,ij} < 1$}\\
\end{array} \right.
\]

\[
F_{(j,A)}^{II} = \left\{ 
\begin{array}{l l}
    \rev{\mathsf{argmin}}_{F' \subseteq \F : \sum_{i \in F'} \overline{x}^{(2)}_{A,ij} \geq 1 } \rev{\mathsf{max}}_{i \in F'} c_{ij} &  
    \mbox{if $\sum_{i \in \F} \overline{x}^{(2)}_{A,ij} \geq 1$}\\
  \emptyset &  \mbox{if $\sum_{i \in \F} \overline{x}^{(2)}_{A,ij} < 1$}\\
\end{array} \right.
\]

Note that  these sets can easily be computed by considering facilities in an order of non-decreasing distances $c_{ij}$
to the considered client $j$.
Since we can split facilities, w.l.o.g., for all $j \in \C$ we assume that if $F_{(j,A)}^{I}$ is nonempty then
$\sum_{i \in F_{(j,A)}^{I}} \overline{x}^{(1)}_{A,ij} = 1$, and if $F_{(j,A)}^{II}$ is not empty then
$\sum_{i \in F_{(j,A)}^{II}} \overline{x}^{(2)}_{A,ij} = 1$.
Define $d_{(j,A)}^I = \rev{\mathsf{max}}_{i \in F_{(j,A)}^I} c_{ij}$ and $d_{(j,A)}^{II} = \rev{\mathsf{max}}_{i \in F_{(j,A)}^{II}} c_{ij}$.
Let $d_{(j,A)} = \rev{\mathsf{min}}\{d_{(j,A)}^I, d_{(j,A)}^{II}\}$.

For a client-scenario pair $(j,A)$, if we have $d_{(j,A)} = d_{(j,A)}^I$, then we call such a pair \emph{first-stage clustered},
and put its \emph{cluster candidate} $F_{(j,A)} = F_{(j,A)}^I$. Otherwise, if $d_{(j,A)} = d_{(j,A)}^{II} < d_{(j,A)}^{I}$, 
we say that $(j,A)$ is \emph{second-stage clustered} and put its \emph{cluster candidate} $F_{(j,A)} = F_{(j,A)}^{II}$

Recall that we use $C_{(j,A)}=\sum_i c_{ij} x^*_{A,ij}$ to denote the fractional connection cost of client $j$ in scenario $A$.
Let us now argue that distances to facilities in \emph{cluster candidates} are not too large. 

\begin{lemma} \label{lemma:d_j_A}
 $d_{(j,A)} \leq \frac{\gamma}{\gamma-2} \cdot C_{(j,A)}$ for all pairs $(j,A)$. 
\end{lemma}

\begin{proof}
Fix a client-scenario pair $(j,A)$. Assume $F_{(j,A)} = F_{(j,A)}^I$ (the other case is symmetric).
Recall that in this case we have $d_{(j,A)} = d_{(j,A)}^I \leq d_{(j,A)}^{II}$.
Consider the following two subcases.\\
\textbf{Case 1.} $\sum_{i \in F_{(j,A)}^{II}} \overline{x}^{(2)}_{A,ij} = 1$.\\
Observe that we have $c_{ij} \geq d_{(j,A)}$ for all $i \in F'=\F\setminus (F_{(j,A)}^I \cup F_{(j,A)}^{II})$.
Note also that $\sum_{i \in F'} \overline{x}_{A,ij} = \gamma - 2$ and
$\sum_{i \in F'} x^*_{A,ij} = \frac{\gamma - 2}{\gamma}$. 
Hence, $C_{(j,A)}=\sum_{i \in \F} x^*_{A,ij} c_{ij} \geq \sum_{i \in F'} x^*_{A,ij} c_{ij} \geq \frac{\gamma - 2}{\gamma} \cdot d_{(j,A)}$.

\noindent
\textbf{Case 2.} $\sum_{i \in F_{(j,A)}^{II}} \overline{x}^{(2)}_{A,ij} = 0$, which implies that 
$\sum_{i \in \F} \overline{x}^{(2)}_{A,ij} < 1$. Observe that now we have $\sum_{i \in \F} \overline{x}^{(1)}_{A,ij}  > \gamma - 1$,
and therefore $\sum_{i \in \F \setminus F_{(j,A)}^I} \overline{x}^{(1)}_{A,ij} > \gamma-2$.
Recall that $c_{ij} \geq d_{(j,A)}$ for all $i \in (\F\setminus F_{(j,A)}^I) $ with $\overline{x}^{(1)}_{A,ij} > 0$, hence
$C_{(j,A)} = \sum_{i \in \F} x^*_{A,ij} c_{ij} \geq \sum_{i \in (\F\setminus F_{(j,A)}^I)\: \overline{x}^{(1)}_{A,ij} > 0 } x^*_{A,ij} c_{ij} > 
\frac{\gamma - 2}{\gamma} \cdot d_{(j,A)}$. 
\end{proof}
 
Like in Section~\ref{sec:LP-rounding}, the algorithm opens facilities randomly in each of the stages
with the probability of opening facility $i$ equal to $\overline{y}_i$
in Stage I, and $\overline{y}_{A,i}$ in Stage II of scenario $A$.
Some facilities are grouped in disjoint \emph{clusters} in order to correlate 
the opening of facilities from a single cluster.
The clusters are formed in each stage by the following procedure.
Let all facilities be initially unclustered.
In Stage I, consider all first-stage clustered client-scenario pairs, 
i.e., pairs $(j,A)$ such that $d_{(j,A)} = d_{(j,A)}^I$
(in Stage II of scenario $A$, consider all second-stage clustered client-scenario pairs)
in the order of non-decreasing values $d_{(j,A)}$.
If the set of facilities $F_{(j,A)}$ contains no facility
from the previously formed clusters, 
then form a new cluster containing facilities from $F_{(j,A)}$,
otherwise do nothing.
In each stage, open exactly one facility in each cluster.
Recall that the total fractional opening of facilities in each cluster
equals $1$. Within each cluster choose the facility randomly with the probability
of opening facility $i$ equal to the fractional opening $\overline{y}_i$ 
in Stage I, or $\overline{y}_{A,i}$ in Stage II of scenario $A$.
For each unclustered facility $i$, open it independently with 
probability $\overline{y}_i$ in Stage I, and with probability $\overline{y}_{A,i}$ in Stage II of scenario $A$.
Finally, at the end of Stage II of scenario $A$, connect each client $i \in A$ to the 
closest open facility.

\paragraph{Analysis.}
The expected facility-opening cost is obviously $\gamma$ times the fractional opening cost.
More precisely, the expected facility-opening cost in scenario $A$ equals $\gamma \cdot F^*_A =
\gamma \cdot \sum_{i \in \mathcal{F}} f_i^I y_i +\sum_i f_i^A y_{A,i}$.
It remains to bound the expected connection cost in scenario $A$ in terms of 
$C^*_A =\sum_{j \in A} \sum_i c_{ij} x_{A,ij}$. 

\begin{lemma}
The expected connection cost of client $j$ in scenario $A$ is at most $(1+\frac{2\gamma+2}{\gamma-2} e^{-\gamma}) \cdot C_{(j,A)}$.
\end{lemma}

\begin{proof}
Consider a single client-scenario pair $(j,A)$.
Observe that the facilities fractionally connected to $j$ in scenario $A$
have the total fractional opening of $\gamma$ in the scaled facility-opening vector $\overline{y}$.
Since there is no positive correlation (only negative correlation in the disjoint clusters formed 
by the algorithm), with probability at least $1-e^{-\gamma}$ at least one such facility will be opened,
moreover, by Lemma~\ref{exp_distance_lemma} the expected distance to the closest of the open facilities from this set
will be at most the fractional connection cost $C_{(j,A)}$.

Just like in the proof of Lemma~\ref{primal_dual_lemma}, from the greedy construction of the clusters
in each phase of the algorithm, with probability $1$, there exists facility $i$ opened 
by the algorithm such that $c_{ij} \leq 3 \cdot d_{(j,A)}$.
We connect client $j$ to facility $i$ if no facility from facilities fractionally serving $(j,A)$ was opened.
We obtain that the expected connection cost of client $j$ is at most $(1-e^{-\gamma}) \cdot C_{(j,A)} + e^{-\gamma} \cdot 3 d_{(j,A)}$.
By Lemma~\ref{lemma:d_j_A}, this can by bounded by  
$(1-e^{-\gamma}) \cdot C_{(j,A)} + e^{-\gamma} \cdot 3 \cdot \frac{\gamma}{\gamma-2} \cdot C_{(j,A)} =
(1+\frac{2\gamma+2}{\gamma-2} e^{-\gamma}) \cdot C_{(j,A)}$
\end{proof}

To equalize the opening and connection cost approximation ratios we solve $(1+\frac{2\gamma+2}{\gamma-2} e^{-\gamma}) = \gamma$ and obtain the following.

\begin{theorem} \label{thm_rand_per-scenario}
The described algorithm with $\gamma = 2.4957$ delivers solutions such that the expected cost in each scenario $A$
is at most $2.4957$ times the fractional cost in scenario $A$. 
\end{theorem}

\paragraph{A note on the rounding procedure.}
An alternative way of interpreting the rounding process that we described above
is to think of facilities as having distinct copies for the first stage 
opening and for each scenario of the second stage. In this setting
each client-scenario pair becomes a client that may only connect itself
to either a first stage facility or a second stage facility of the specific scenario.
This results in a specific instance of the standard Uncapacitated Facility Location problem,
where triangle inequality on distances holds only within certain metric subsets of locations.
One may interpret our algorithm as an application of a version of the algorithm of Chudak and Shmoys~\cite{ChudakS98}
for the UFL problem applied to this specific UFL instance, where we take special care to create clusters only from
facilities originating from a single metric subset of facilities. Such a choice of the clusters
is sufficient to ensure that we may connect each client via a 3-hop path to a facility opened in each cluster. 

\subsection{Strict per scenario bound}
\label{sec:strict-per-scenario}

Note that our algorithm is a randomized algorithm with
bounded expected cost of the solution. The opening of facilities
in the second stage of the algorithm can be derandomized
by a standard technique of conditional expectations.
However, by simply derandomizing the first stage we would
obtain a solution for which the final cost could no longer be bounded as in Theorem~\ref{thm_rand_per-scenario} for 
every single scenario A. 
Nevertheless, a weaker but deterministic bound on the connection cost of each client is still possible.
Recall that by Lemma~\ref{lemma:d_j_A} we have that every client-scenario pair $(j,A)$ 
client $j$ is certainly connected with cost at most $d_{(j,A)} \leq \frac{\gamma}{\gamma-2} \cdot C_{(j,A)}$.
Hence, our algorithm with $\gamma=2.4957$ has a deterministic guarantee that each client is connected
in the final solution at most $15.11$ times further than in the fractional solution.

Please note that all the bounds are only based on the primal feasible solution,
and hence it is possible to use this approach to add any specific limitations 
as constraints to the initial linear program. In particular one may introduce budgets
for scenarios, or even to restrict the connection cost of a selected subset of clients.

There is a natural trade-off between the deterministic bounds
and bounds on the expectations. By selecting a different value of the initial 
scaling parameter we obtain different points on the trade-off curve.
In particular, for $\gamma=5$ we expect to pay for facility opening $5$ times the fractional opening cost,
the expected connection cost in each scenario will be $1.027$ times the fractional connection cost in this scenario,
and the deterministic upper bound on the connection cost of a single client will be at most $5$ times
the fractional connection cost of this client.

There remains the issue of nondeterminism of facility opening costs.
Observe that the second stage openings may easily be derandomized
by the standard conditional expectations method. In the derandomization process
it is possible to guarantee that a single linear objective is no more then it's expectation.
It is hence possible to have such guarantee, e.g., for the Stage II facility opening plus connection cost.

As mentioned before, it is not trivial to derandomize the Stage I facility openings.
If we decide for any fixed openings, we no longer have the per-scenario bound on the expected connection cost. 
Typically, a fixed Stage I decision favors one scenario over another.
Nevertheless, taking the most adversarial point of view, we may decide for $\gamma=5$, and then
for the cheapest possible realization of Stage I facility openings, still the connection costs will be at most $5$ 
times the fractional connection costs. This results in a truly deterministic $5$ approximation algorithm
in the strongest per-scenario model.

\smallskip
\noindent \textbf{Acknowledgments.} We thank Alok Baveja, David Shmoys and 
Chaitanya Swamy for helpful discussions and inputs. Thanks again to 
David and Chaitanya for writing the survey
\cite{SwamyS06} that inspired this work. 
Our gratitude also to Meena Mahajan and the
Institute of Mathematical Sciences at Chennai, India, for their
hospitality; part of this writing was completed at the IMSc. 
We also thank the anonymous reviewers who contributed greately
to the clarity of the presented models and algorithms.

The work of J. Byrka was partially conducted at CWI in Amsterdam,  
TU Eindhoven, EPFL Lausanne, and during a visit at the University of Maryland. 
His research was partially supported by the Future and Emerging Technologies 
Unit of EC (IST priority - 6th FP), under contract no. FP6-021235-2 (project ARRIVAL);
MNiSW grant number N N206 368839; FNP Homing Plus program, contract no. Homing Plus/2010-1/3; and Polish National Science Centre grant 2015/18/E/ST6/00456.
A. Srinivasan was partially supported by NSF Award CCR 0208005, 
NSF ITR Award CNS 0426683, NSF Awards CNS 0626636, CNS 1010789, CCF-1422569, and CCF-1749864 and by a research award from Adobe, Inc.

\end{document}

%% file: preamble-new.tex
\newtheorem{theorem}{Theorem}[section]
\newtheorem{lemma}[theorem]{Lemma}
\newtheorem{corollary}[theorem]{\indent Corollary}

\newcommand{\qed}{\nopagebreak \hfill $\Box$}
\newenvironment{proof}{\par \noindent {\em Proof:}}{\qed \par}

  {\begin{center}%
   \begin{minipage}[b]{0.8\textwidth}
   \begin{list}{}%
    {\setlength{\listparindent}{-2.5mm}%
     \setlength{\itemindent}{\listparindent}}%
     \setlength{\labelsep}{0in}%
     \setlength{\parskip}{0ex}%
     \setlength{\parsep}{\parskip}%
     \setlength{\topsep}{0in}%
     \setlength{\partopsep}{0in}%
   \item[]%
  }%
  {\end{list}%
   \end{minipage}%
   \end{center}%
  }
 
  {\begin{list}{}%
    {\setlength{\leftmargin}{5mm}%
     \setlength{\listparindent}{-0.5\leftmargin}%
     \setlength{\itemindent}{\listparindent}}%
     \setlength{\labelsep}{0in}%
     \setlength{\parskip}{0ex}%
     \setlength{\parsep}{\parskip}%
     \setlength{\topsep}{0in}%
     \setlength{\partopsep}{0in}%
    \item[]%
  }%
  {\end{list}%
  }
 


%% file: spage.tex
\newcommand{\ol}{\setlength{\itemsep}{0pt.}\begin{enumerate}}
\newcommand{\eol}{\end{enumerate}\setlength{\itemsep}{-\parsep}}

%% file: StochOptArxiv.bbl
\begin{thebibliography}{35} 

\bibitem{Beale55} 
E.~M.~L. Beale. 
\newblock On minimizing a convex function subject to linear inequalities.
\newblock {\em Journal of the Royal Statistical Society, Series B}, 17:173--184;
  discussion 194--203, 1955. 
 
\bibitem{BirgeL97} 
J.~R. Birge and F.~V. Louveaux.
\newblock {\em Introduction to Stochastic Programming}.
\newblock Springer-Verlag, NY, 1997.

\bibitem{ByrkaAardal10}
J. Byrka and K. Aardal.
\newblock An Optimal Bifactor Approximation Algorithm for the Metric Uncapacitated Facility Location Problem
\textit{SIAM J.~Comput.}, 39(6):2212--2231, 2010. 
 
\bibitem{CharikarCP05} 
M.~Charikar, C.~Chekuri, and M.~P\'{a}l.
\newblock Sampling bounds for stochastic optimization. 
\newblock {\em Proceedings, 9th RANDOM}, pages 257--269, 2005.

\bibitem{ChudakS98}
F.~Chudak and D.~Shmoys.
\newblock Improved approximation algorithms for the uncapacitated facility
  location problem.
\newblock {\em {SIAM} Journal on Computing}, 33(1):1--25, 2003.
 
\bibitem{Dantzig55} 
G.~B. Dantzig. 
\newblock Linear programming under uncertainty.
\newblock {\em Mgmt. Sc.}, 1:197--206, 1955. 
 
\bibitem{DeanGV04}
B.~Dean, M.~Goemans, and J.~Vondrak.
\newblock Approximating the stochastic knapsack problem: the benefit of adaptivity.
\newblock {\em Mathematics of Operations Research}, 33(4): 945--964, 2008.
 
\bibitem{DhamdhereGRS05} 
K.~Dhamdhere, V.~Goyal, R.~Ravi, and M~Singh.
\newblock How to pay, come what may: approximation algorithms for demand-robust covering
  problems. 
\newblock {\em Proceedings, 46th Annual {IEEE} Symposium on Foundations of Computer 
  Science}, pages 367--378, 2005. 
 
\bibitem{DyerS05} 
M.~Dyer and L.~Stougie. 
\newblock Computational complexity of stochastic programming problems. 
\newblock {\em Mathematical Programming}, 106(3): 423--432, 2006.

\bibitem{FKG} 
C.M.~Fortuin, P.W.~Kasteleyn, J.~Ginibre.
\newblock Correlation inequalities on some partially ordered sets.
\newblock {\em Communications in Mathematical Physics} {22}, pages 89�-103, 1971.

\bibitem{GuhaK99}
S.~Guha and S.~Khuller.
\newblock Greedy strikes back: Improved facility location algorithms.
\newblock {\em Journal of Algorithms}, 31(1):228--248, 1999.
 
%

\bibitem{GuptaPRS11} 
A.~Gupta, M.~P\'{a}l, R.~Ravi, and A.~Sinha.
\newblock   Sampling and cost-sharing: Approximation algorithms for stochastic optimization problems.
\newblock {\em SIAM Journal on Computing}, 40(5): 1361--1401, 2011.
 
  
\bibitem{GuptaRS07}
A.~Gupta, R.~Ravi, and A.~Sinha.
\newblock LP rounding approximation algorithms for stochastic network design. 
\newblock {\em Mathematics of Operations Research}, 32(2): 345--364, 2007.

\bibitem{ImmorlicaKMM04} 
N.~Immorlica, D.~Karger, M.~Minkoff, and V.~Mirrokni. 
\newblock On the costs and benefits of procrastination: approximation algorithms for
  stochastic combinatorial optimization problems.
\newblock {\em Proceedings, 15th Annual {ACM}-{SIAM} Symposium on Discrete
  Algorithms}, pages 684--693, 2004.

\bibitem{JainMMSV02}
K.~Jain, M.~Mahdian, E.~Markakis, A.~Saberi, and V.~Vazirani.
\newblock Greedy facility location algorithms analyzed using dual-fitting with 
  factor-revealing LP.
\newblock {\em Journal of the {ACM}}, 50(6):795--824, 2003.

\bibitem{JainMS02}
K.~Jain, M.~Mahdian, and A.~Saberi.
\newblock A new greedy approach for facility location problems.
\newblock In {\em Proceedings of the 34th Annual {ACM} Symposium on Theory of
  Computing}, pages 731--740, 2002.

\bibitem{KleywegtSH01} 
A.~J. Kleywegt, A.~Shapiro, and T.~Homem-De-Mello. 
\newblock The sample average approximation method for stochastic discrete optimization.
\newblock {\em SIAM Journal of Optimization}, 12:479--502, 2001. 

\bibitem{Li11}
S. Li.
\newblock A 1.488 Approximation Algorithm for the Uncapacitated Facility Location Problem.
\newblock In {\em Proceedings of the 38-th ICALP}, pages 77-88, 2011.


\bibitem{DBLP:conf/stoc/LinV92}
J.-H. Lin and J. S. Vitter.
\newblock epsilon-Approximations with Minimum Packing Constraint Violation (Extended Abstract)
\newblock {\em Proceedings, 24th Annual {ACM} Symposium on Theory of Computing}, pages 771--782, 1992.
 
\bibitem{LinderothSW04} 
J.~Linderoth, A.~Shapiro, and R.~K. Wright.
\newblock The empirical behavior of sampling methods for stochastic programming.
\newblock {\em Annals of Operations Research}, 142(1):215--241, 2006.
 
\bibitem{MahdianYZ05}
M.~Mahdian, Y.~Ye, and J.~Zhang.
\newblock Approximation algorithms for metric facility location problems.
\newblock \textit{SIAM Journal on Computing}, 36(2):411-432, 2006. 
 
\bibitem{MohringSU99} 
R.~M\"{o}hring, A.~Schulz, and M.~Uetz.
\newblock Approximation in stochastic scheduling: the power of LP based priority
  policies. 
\newblock {\em JACM}, 46:924--942, 1999. 

\bibitem{naor-roth}
M.~Naor and R.~M. Roth.
Optimal file sharing in distributed networks.
\textit{SIAM J.~Comput.}, 24:158--183, 1995. 

\bibitem{ps:edge-col}
A.~Panconesi and A.~Srinivasan.
Randomized distributed edge coloring via an extension of the
Chernoff-Hoeffding bounds.
\textit{SIAM J.~Comput.}, 26:350--368, 1997. 

\bibitem{raghavan-thompson}
P.~Raghavan and C.~D.\ Thompson.
Randomized rounding: a technique for provably good algorithms and
algorithmic proofs.
\textit{Combinatorica}, 7:365--374, 1987. 

\bibitem{RaviS04} 
R.~Ravi and A.~Sinha. 
\newblock Hedging uncertainty: approximation algorithms for stochastic optimization problems.
\newblock {\em Mathematical Programming}, 108(1): 97-114, 2006.
 
\bibitem{RuszczynskiS03} 
A.~Ruszczynski and A.~Shapiro. 
\newblock Editors, {\em Stochastic Programming}, volume 10 of {\em Handbooks in Operations                                                                                        
  Research and Mgmt. Sc.}, North-Holland, Amsterdam, 2003.
 
\bibitem{Shapiro03}
A.~Shapiro. 
\newblock Monte Carlo sampling methods. 
\newblock In A.~Ruszczynski and A.~Shapiro, editors, {\em Stochastic Programming}, volume
  10 of {\em Handbooks in Operations Research and Mgmt. Sc.}, North-Holland,
  Amsterdam, 2003.
 
\bibitem{ShapiroN04} 
A.~Shapiro and A.~Nemirovski. 
\newblock On complexity of stochastic programming problems. 
\newblock Published electronically in {\em Optimization Online}, 2004.
\newblock {\ttfamily\small http://www.optimization- online.org/DB\_FILE/2004/10/978.pdf}.
 
\bibitem{ShmoysS04} 
D.~B. Shmoys and C.~Swamy. 
\newblock An approximation scheme for stochastic linear programming and its application to
  stochastic integer programs.
\newblock {\em JACM}, 53(6):978--1012, 2006. 
\newblock Preliminary version appeared as ``Stochastic optimization is (almost) as easy as                                                                                        
  deterministic optimization'' in {\em Proceedings, 45th Annual {IEEE} {FOCS}}, pages
  228--237, 2004. 

\bibitem{ShmoysTA97}
D.~B. Shmoys, {\'E}.~Tardos, and K.~I. Aardal.
\newblock Approximation algorithms for facility location problems.
\newblock In {\em Proceedings of the 29th Annual {ACM} Symposium on Theory of
  Computing}, pages 265--274, 1997.

\bibitem{SoZY06}
A.~M-C. So, J.~Zhang, and Y.~Ye.
Stochastic combinatorial optimization with controllable risk aversion level.
\newblock {\em Mathematics of Operations Research}, 34(3): 522--537, 2009.

\bibitem{srin:pos-correl}
A.~Srinivasan.
Improved approximation guarantees for packing and
covering integer programs.
\textit{SIAM J.\ Comput.}, 29:648--670, 1999. 

\bibitem{Srinivasan07}
A.~Srinivasan.
Approximation algorithms for stochastic and risk-averse optimization.
\newblock {\em Proceedings, 18th SODA}, pages 1305--1313, 2007.

\bibitem{Sviridenko02}
M.~Sviridenko.
\newblock An improved approximation algorithm for the metric uncapacitated
  facility location problem.
\newblock In {\em Proceedings of the 9th International Conference on Integer Programming
  and Combinatorial Optimization}, pages 240--257, 2002.

\bibitem{Swamy04} 
C.~Swamy. 
\newblock {\em Approximation Algorithms for Clustering Problems}. 
\newblock Ph.D. thesis, Cornell University, Ithaca, NY, 2004.

\bibitem{Swamy11}
C.~Swamy. 
\newblock  Risk-averse stochastic optimization: probabilistically-constrained models and algorithms for black-box distributions.
\newblock In {\em Proceedings, 22nd SODA}, pages 1627--1646, 2011.
 
\bibitem{SwamyS05}
C.~Swamy and D.~B. Shmoys. 
\newblock Sampling-based approximation algorithms for multi-stage stochastic
  optimization. 
\newblock In {\em  SIAM J.\ Comput.} 41(4): 975--1004, 2012.

\bibitem{SwamyS06}
C.~Swamy and D.~B. Shmoys. 
\newblock Algorithms Column: Approximation Algorithms for 2-stage Stochastic
Optimization Problems.
\newblock \textit{SIGACT News}, 37:33--46, 2006. 
 
\bibitem{VerweijAKNS03} 
B.~Verweij, S.~Ahmed, A.~J. Kleywegt, G.~L. Nemhauser, and A.~Shapiro.
\newblock The sample average approximation method applied to stochastic routing problems:
  a computational study.
\newblock {\em Computational Optimization and Applications}, 24:289--333, 2003.
 
\end{thebibliography}
